\documentclass[a4paper,11pt]{article}
\pdfoutput=1 
\usepackage{jheppub} 
\usepackage[T1]{fontenc} 
\usepackage{overpic}

\usepackage{color}
\usepackage{amsthm}
\usepackage{graphicx}
\usepackage{slashed}
\usepackage{amssymb}
\usepackage{tikz}
\usetikzlibrary{arrows}
\usetikzlibrary{decorations.markings}

\def\slashchar#1{\setbox0=\hbox{$#1$}           
   \dimen0=\wd0                                 
   \setbox1=\hbox{/} \dimen1=\wd1               
   \ifdim\dimen0>\dimen1                        
      \rlap{\hbox to \dimen0{\hfil/\hfil}}#1 
   \else                                        
      \rlap{\hbox to \dimen1{\hfil$#1$\hfil}}/                                    \fi}

\newtheorem{prop}{Proposition}
\newtheorem{cor}{Corollary}

\newcommand\hide[1]{{}}

\title{\boldmath Significance Variables}

\author[a]{Benjamin Nachman}
\author[b]{Christopher G.~Lester}

\affiliation[a]{DAMTP, CMS, University of Cambridge, Wilberforce Road, Cambridge, CB3 0HA, U.K.}
\affiliation[b]{Cavendish Laboratory, Department of Physics, JJ Thomson Avenue, Cambridge, CB3 0HE, U.K.}

\emailAdd{bnachman@cern.ch}
\emailAdd{Lester@hep.phy.cam.ac.uk}

\abstract{
Many particle physics analyses which need to discriminate some background process from a signal ignore event-by-event resolutions of kinematic variables.  Adding this information, as is done for missing momentum significance, can only improve the power of existing techniques.  We therefore propose the use of significance variables which combine kinematic information with event-by-event resolutions.  We begin by giving some explicit examples of constructing optimal significance variables.  Then, we consider three applications: new heavy gauge bosons, Higgs to $\tau\tau$, and direct stop squark pair production.  We find that significance variables can provide additional discriminating power over the original kinematic variables: $\sim$ 20\% improvement over $m_T$ in the case of $H\rightarrow\tau\tau$ case, and $\sim$ 30\% impovement over $m_{T2}$ in the case of the direct stop search.  
}

\begin{document} 
\maketitle
\flushbottom

\section{Introduction}

There is a set of key observables which seem, hitherto, to have
received scant to non-existent attention in the 
literature. These observables are the {\em event-by-event}
resolutions of individual kinematic variables which constitute the
building blocks of most analyses at present.  Such analyses (which we will
call ``cut-based'') will, for the foreseeable future, continue to be
found in a large fraction of collider physics search papers, even
though more powerful techniques are available.\footnote{In the
appropriate context, any technique which make sensible and full use
of the joint likelihood of the data as a function of all relevant
parameters cannot be beaten.}  One of the main reasons that
cut-and-count usage remains strong, despite non-optimality, is the
perceived simplicity with which ``reasonable'' analyses can be
developed.  Against this backdrop we should ask: {{\em ``How can
event-by-event resolutions be used effectively within current
analyses without fundamentally changing the way they are
done?''}

\section{A concrete example}

Consider a kinematic variable $m$ which, in the absence of new physics and detector resolutions, has a classical maximum $M$.  For example, $m$ could be transverse momentum or the actual mass of some system of particles.   The usual procedure for using $m$ is to place a cut value $m_{cut}$ and then to count the number of events for which $m>m_{cut}$.  If this number significantly exceeds expectation, then one has evidence for new physics.  However, one can do better than this by including more information such as event-by-event resolutions (and the mass scale $M$).  For example, consider the probability $P_M$ that the measured value $m^{\text{observed}}$ for a fixed event exceeds the scale $M$.  Symbolically, this is 

\begin{align}
\label{defofP}
P_M=\Pr(m^{\text{(re)measured}}>M|R_m),
\end{align}

\noindent where $R_m$ is the resolution function\footnote{$p$ will be the generic symbol for a probability density function.} $p(m^{\text{(re)measured}}|m^{\text{observed}})$.  For general purposes, one assumes that $R_m$ is a Gaussian function centered at the measured value with a width given by $\sigma_m$.  In this case, we can explicitly compute $P_M$, as in Eq.~\ref{IntroduceX}.

\begin{align}
\label{IntroduceX}
P_M&=\int_M^\infty p(m^{\text{(re)measured}}|R_m)dm^{\text{(re)measured}}\\\nonumber
&=\frac{1}{\sqrt{2\pi\sigma_m}}\int_M^\infty \exp\left(\frac{-(m^{\text{(re)measured}}-m^{\text{observed}})^2}{2\sigma_m^2}\right)dm^{\text{(re)measured}}\\\nonumber
&=\frac{1}{2}\left(1+\text{erf}\left(\frac{m^{\text{observed}}-M}{\sigma_m\sqrt{2}}\right)\right).
\end{align}

\noindent Since the $\mathrm{erf}$ function is monotonic and smooth, the complete behavior of $P_M$ is determined by the quantity 

\begin{align}
X_M\equiv \frac{m^{\text{observed}}-M}{\sigma_m}.
\end{align}

\par

Perhaps surprisingly, very few analyses seem to use quantities like $X_M$.   In fact, so far as the authors are aware, the only variable of this type that has seen significant usage in the collider literature is the ``$E_T^{\text{miss}}$ significance'', not to be confused with $E_T^{\text{miss}}$.  The latter is the magnitude of the transverse momentum necessary for conservation in the plane perpendicular to the beam whereas $E_T^{\text{miss}}$ significance, first constructed at D\O~\cite{D0}, in its most complete form usually refers to the log of a likelihood ratio
\begin{align}
\label{METsig}
\log\left(\frac{p(\slashed{E}_T^{\text{}}=\slashed{E}_T^{\text{measured}})}{p(\slashed{E}_T^{\text{}}=0)}\right),
\end{align}
where $p(\slashed{E}_T^{\text{}}=x)$ is the probability density for remeasured valued of the missing transverse energy.  The purpose of $E_T^{\text{miss}}$ significance is to differentiate events with real missing energy from invisible particles like neutrinos from those without, and it is constructed from the resolution functions of all the objects used to construct the $E_T^{\text{miss}}$ itself.

For Gaussian resolutions, the $E_T^{\text{miss}}$ significance is a monotonic function of  $\left(\slashed{E}_T^{\text{measured}}\right)^2/2\sigma_{\slashed{E}_T}^2$.  In general, it can be tedious to precisely determine $\sigma$ on an event-by-event basis.  Therefore, one observes~\cite{ScalingATLAS,ScalingCMS} that $\sigma_{\slashed{E}_T}\propto \sqrt{H_T}$, the scalar sum of the visible $p_T$ in the event.  Then, an approximate $E_T^{\text{miss}}$ significance may be written as a monotonic function of $(E_T^{\text{miss}})^2/{H_T}$ and in fact, the most commonly used choice is $E_T^{\text{miss}}/\sqrt{H_T}$.

We note that the approximate $E_T^{\text{miss}}$ significance defined above {\em is} a realisation of $X_M$ in which (i) $M=0$, (ii) we assume a Gaussian resolution function centered at the measured $E_T^{\text{miss}}$, and (iii) $\sigma\propto\sqrt{H_T}$.

Even though $E_T^{\text{miss}}/\sqrt{H_T}$ and $E_T^{\text{miss}}$ and
are correlated, one can gain statistical power by considering
$E_T^{\text{miss}}/\sqrt{H_T}$ in addition to or instead of
$E_T^{\text{miss}}$ itself.  This has been shown in analyses spanning
a wide range of physics processes including Standard Model
measurements~\cite{D0Precision,CDFttbarCrossSection,ArielS,CDFAnamalousggET,CDFWWZZ}
and searches for the Higgs Boson~\cite{D0Higgs}, Dark
Matter~\cite{DarkCDF}, and Supersymmetric particles
~\cite{SingleLeptonStopATLAS,D0SUSY}.

Motivated by the gains found by using the missing energy significance
$E_T^{\text{miss}}/\sqrt{H_T}$ in addition to $E_T^{\text{miss}}$, we
want to see whether similar profits are to be had from building
significance related quantities for other kinematic variables.


\section{Significance variables}

There are many ways that cut-based analyses could be modified to make good use of event-by-event resolutions.  The least prescriptive (and in some cases least effective) method simply adds to each event the resolutions as additional variables in their own right upon which to make cuts.  Indeed, simply doing this and leaving a Multivariate Analysis (MVA) tool to find the best way of using the additional information will appeal to many.\footnote{It is straightforward to show (see Appendix~\ref{sec:proofsandstuff}) that the optimal way of making use of the information in a cut-based analysis is always equivalent to a cut on the ratio of the likelihoods of the event under the signal and background hypotheses, and MVA tools can often get pretty close to such cuts.}

However, readers will have noted that the physics of the preceding example of $E_T^{\text{miss}}$ significance motivated the formation of a very particular combination of the kinematic variable and its associated resolution into a single quantity, equivalent to the significance variable $X_M$, which may contain all of the relevant discriminatory information.   We would like to show that it is not unusual for most of the relevant resolution information to be condensed into a single simple $X$-like variable.  Furthermore, we will show that it is even commonplace under certain conditions -- principally those in which the signal and backgrounds are associated with different mass or energy scales.

Knowing that variables like $X_M$ frequently contain most of the relevant resolution information is  useful.  It means that a user keen to see whether an analysis can benefit from incorporating resolution information has a straightforward way of testing whether it might help.  For each event, using the description below, one can compute a $X_M$ significance variable for the kinematic quantity of interest, and then try placing a cut on  $X_M$ instead of (or perhaps in addition to) the cut on the kinematic variable on which his $X_M$ was based.

If it is desired to include resolution information in an analysis, the work necessary to compute that resolution an any particular kinematic variable is unavoidable, and specific to the analysis in question.  However it is important to note (i) that this work is the same regardless of whether the resolution be used in an MVA or in the construction of an $X_M$-like significance variable, and (ii) that the construction of an $X_M$-like significance variable is itself very simple, requiring only a subtraction, a division and the choice of a signal-background separation scale $M$.  Given that $X_M$-like variables are frequently close to optimal (as we show below) there seems little reason to avoid adding them to our toolkits.

Finally, before moving on to specific examples, we not the $X_M$ itself will not {\em always} be the optimal significance variable for an analysis.  Any case in which resolutions are significantly non-Gaussian may require, for optimality, the use of a significance variable based on the likelihood ratio as described in Appendix~\ref{sec:proofsandstuff}, or the use of an MVA tool to approximate the likelihood ratio procedure.  Nonetheless, our key message is that many analyses could make use of resolution information at the event-by-event level which they are presently throwing away, and that even if they do nothing else, analyses should consider using this information.  A simple way of using it, that captures most of the information thrown away is contained in an $X_M$-like significance variable, but where this is non-optimal, the resolution information can and should still be used either with an MVA or a dedicated derivation of the optimal significance variable(s) for the analysis in question.

\section{Some worked examples of {\em optimal} significance variables in {\em toy} models}

\subsection{The simplest case of all -- Gaussian resolution}

Consider a search for a physics processes using a single kinematic variable $m$.  Using the significance metric $\hat{s}(c)\equiv s/\sqrt{b}$, for $c$ a cut value,  we can ask the question how does $\max_c\hat{s}$ change if we also include some measure of the resolution on $m$?  In other words, what is the optimal combination of $m$ and $\sigma_m$ to maximize the significance metric $\hat{s}$?  To begin, consider a simple model in which the variable $m$ has a delta function distribution, $(1/N)dm_i/dN=\delta(m-M_i)$, where $i\in\{s,b\}$ (signal/background).  For example, suppose that $m=m_T$ in a class search for a heavy gauge boson in the letpon+missing energy channel.   Due to the Jacobian peak, most of the probability for $m$ is near $M_i$, and so this simple model may capture some aspects of the analysis.   Let the resolution functions of $m$ be Gaussian with width $\sigma$.  Then,
\begin{align}
\label{gaus}
p_i(m,\sigma)=g(\sigma)\frac{1}{\sqrt{2\pi\sigma^2}}\exp\left(-\frac{(m-M_i)^2}{2\sigma^2}\right),
\end{align}
where $g(\sigma)$ is the distribution of $\sigma$.  We assume that $g$ is not a delta function, otherwise the resolution information does not tell us anything.   For the reasons set out in Appendix~\ref{sec:proofsandstuff}, the optimal cut boundary on a combination of $m$ and $\sigma$ is a cut on the ratio $p_s(m,\sigma)/p_b(m,\sigma)$.  Dividing the probably functions from above and monotonically transforming the answer brings us to the conclusion that an appropriately chosen cut on the significance variable
\begin{align}
V_\text{opt}^\text{(Gaussian)} = \frac{m-(M_s+M_b)/2}{\sigma^2}
\end{align}
cannot be beaten.  We note that this significance variable is very similar to $X_M$ (with $M=(M_s+M_b)/2$) and only differs in the use of the variance instead of the standard deviation of the uncertainty in the denominator.

\subsection{More realistic asymmetric resolutions}

We now consider a variant of the previous example.  Up until now, we have studied only symmetric resolution smearing.  However, due to falling prior kinematic spectra, more generally we might expect asymmetric resolution functions.  Consider for example a Gumbel distribution for the resolution function:
\begin{align}
p_i(m)=\frac{1}{\beta}\exp\left(\frac{m-M_i}{\beta}\right)\exp\left(-\exp\left(\frac{m-M_i}{\beta}\right)\right)
\end{align}
We choose this probability density function because with the identification $\sigma=\frac{e}{\sqrt{2\pi}}\beta$, to second order in the Taylor expansion, the Gumbel and the Gaussian are the same.  The asymmetry in the Gumbel then is present at the third order.  In the above parameterization, the tail for the Gumbel is heavier on the left than the right, which represents the generic case in which events are more likely to have smeared from lower values due to falling priors.  As we saw in the previous example, it does not matter what $\beta$ weighting function we add to multiply $p_i$ by, so long as it does not depend on $i$, and this time we find that an appropriately chosen cut on 
\begin{align}
V_\text{opt}^\text{(Gumbel)} =\exp\left(\frac{m-M_b}{\beta}\right)-\exp\left(\frac{m-M_s}{\beta}\right)+\frac{M_b-M_s}{\beta}
\end{align}
cannot be bettered for discrimination of signal from background in this model.

\par

The lines of constant $p_s/p_b$ (equivalently the lines of constant $V_\text{opt}^\text{(Gumbel)}$) are richer than for the Gaussian case. In the uninteresting case where $m\ll M_b$ (and thus also $m\ll M_s$ as we will assume, without loss of generality, that there is a hierarchy of scales $M_s>M_b$), we have that the uncertainty parameter $\beta$ is the optimal cut value (i.e. $m$ does not give any information).  Since one looks at counts which exceed bounds, we are interested more in the kinematic maxima and thus when $m\sim M_i$ and when $m> M_i$.  If $M_b<m<M_s$, then the expression above reduces to the variable $X$ with $M=M_b$.  Likewise, if $m>M_s$ and $\beta$ is small compared $M_s-M_b$.   For $m>M_s$ and $\beta$ small compared $M_s-M_b$, both exponentials are large and we can reduce the expression to
\begin{align}
\exp\left(\frac{m-\bar{M}}{\beta}\right) \sinh\left(\frac{M_s-M_b}{\beta} \right)=\mathrm{constant}
\end{align}
where $\bar{M}$ is the average of $M_s$ and $M_b$.   The $\sinh$ term is relatively smaller and slowly varying and thus this is simply $X$ with $M=\bar{M}$.  Figure~\ref{Gumbel} shows a plot of $p_s(m,\beta)/p_b(m,\beta)$ for $M_b=80$ and $M_s=85$.  The level sets of Figure~\ref{Gumbel} correspond to the optimal combination of $m$ and $\beta$.  Straight lines indicate that $X$ is the optimal variable.   One can clearly see that for $m>M_b$, the level sets are straight lines and thus some form of $X$ is optimal.  

\begin{figure}
\begin{center}
\includegraphics[width=0.4\textwidth]{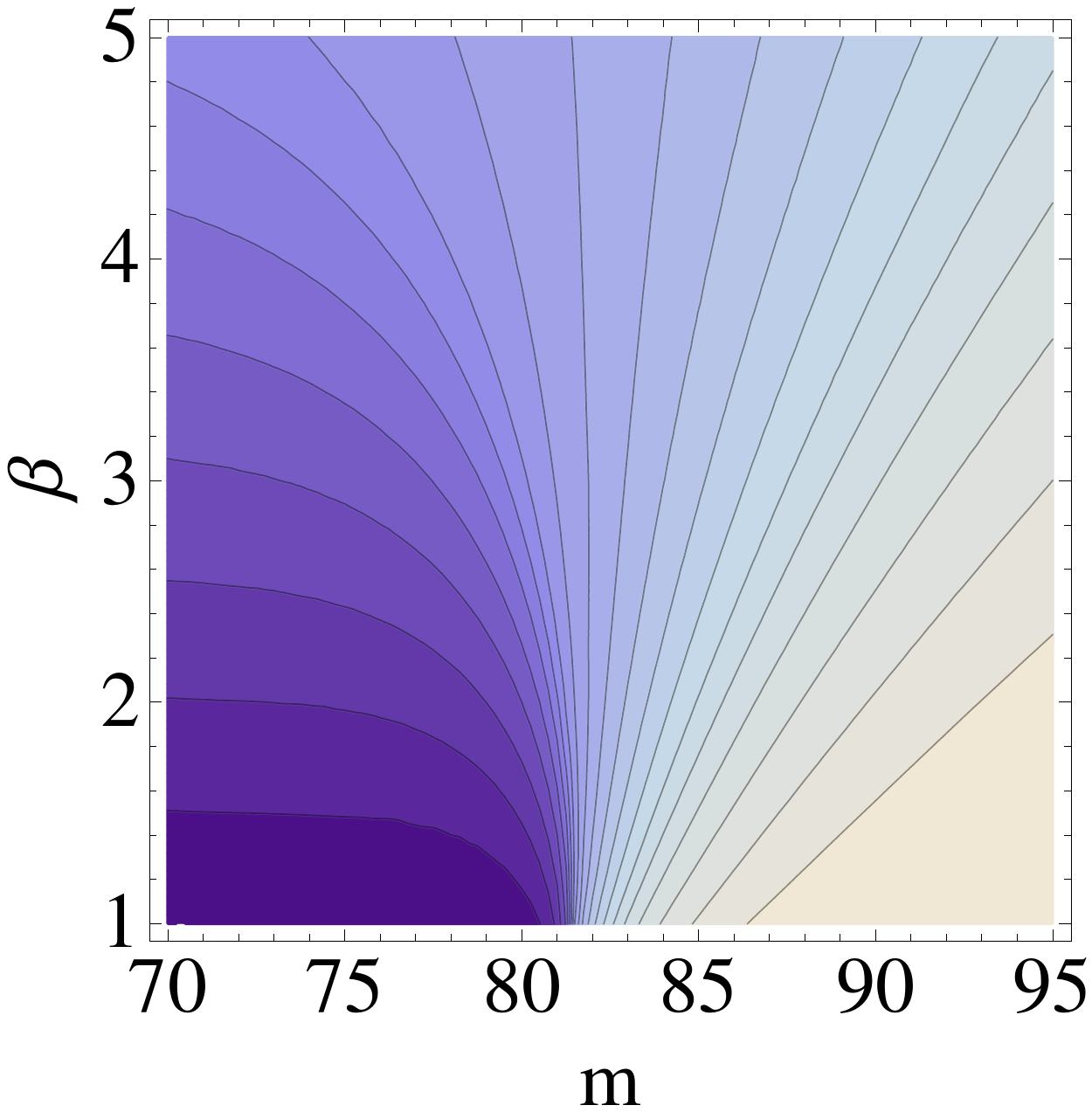}
\end{center}
\caption{Contours of constant $p_s(m,\beta)/p_b(m,\beta)$ (equivalently lines of constant $V_\text{opt}^\text{(Gumbel)}$) in the $(m,\beta)$ plane for $M_b=80$ and $M_s=85$.  We can see that for $m>M_b$ the contours are straight lines and thus $X$ is the optimal variable.}
\label{Gumbel}
\end{figure}

\subsection{Choosing the separation scale $M$}

The above constructions shows that $M$ can play a dynamic role in the definition of $X$.  The interpretation of $M$ as the scale of Standard Model physics does not require that it be fixed ahead of time, since detector resolutions can distort the {\it reconstructed} scale away from the {\it true} scale.  We can further quantify the dependance of $X$ on $M$ by studying the efficacy of $X$ over $m$ with respect to $s/\sqrt{b}$. 

\begin{prop}
The maximum significance for $X_M$, taken over all values of $M$, can be no worse than the maximum significance of $m$ itself.
\end{prop}

\begin{proof}
Suppose tha $k$ is a cut value on $m$ such that $\hat{s}(k)=\max_c\hat{s}$ for $m$.  Then, let $M=k$ and then a cut of $X=0$ will reproduce the same significance as $\hat{s}(k)$.
\end{proof}

\begin{cor}
There is no reason to be afraid of using $X_M$ instead of $m$ since (provided the value of $M$ is chosen sensibly) an $X_M$-only analysis cannot be worse than an than an $m$-only analysis. 
\end{cor}

Now, consider a kinematic variable $m$ with zero resolution maximum $\tilde{m}$.   The value of $M$ which maximizes $\max_{c}\hat{s}_{X(M)}(c)$ need not be equal to $\tilde{m}$.  Obviously, if $\sigma$ is constant over all events, $X$ induces the same ordering on events as $m$ and so any value of $M$ maximizes $\hat{s}$.  Intuitively, it would seem like for varying resolutions, the optimal $M$ should be greater than $\tilde{m}$, but this need not be the case.

 \begin{prop}
Consider a kinematic variable $m$ with zero resolution maximum $\tilde{m}$.  The optimal value of $M$ may be less than $\tilde{m}$.
\end{prop}

\begin{proof}
Consider the model in Eq.~\ref{gaus}.  We know that if the distribution of $\sigma(m)$ is also a delta function, then $X$ and $m$ will give the same significance.  Therefore, take a simple extension:
\begin{align}
g(\sigma) =p\delta(\sigma-\sigma_1)+(1-p)\delta(\sigma-\sigma_2)
\end{align}
where $\sigma_i$ are two fixed values of $\sigma$ and $p\in[0,1]$.  Note that we assume that $\sigma$ is independent of $m$.  With this simple model, we can easily compute the distributions of $m$, $X$ and $\hat{s}$, as seen in Figure~\ref{mstar} for $\tilde{m}=80$ for the background, $\tilde{m}=90$ for the signal, $p=1/2$ and $\rho$ is the signal efficiency, defined by $\rho(c)=\int_c^\infty\mathrm{d}x f(x)$ for $f(x)$ the signal probability density function and $c$ a cut value.  In this setup, we can see that there is an $M<\tilde{m}$ which outperforms the significance at $M=\tilde{m}$.  This is seen clearly in the second plot of the figure in which the low value of $M$ can allow for $X$ to distinguish between low and high resolution events for the signal.  In the limit as $\tilde{m}-M > \sigma$, $X$ will be able to distinguish the low and high resolution events, thus increasing $\hat{s}$.  For $\tilde{m}-M \gg \sigma$, the efficacy of $X$ approaches the constant resolution case and so one cannot gain more by decreasing $M$.
\end{proof}

\begin{figure}
\begin{center}
\includegraphics[scale=0.25]{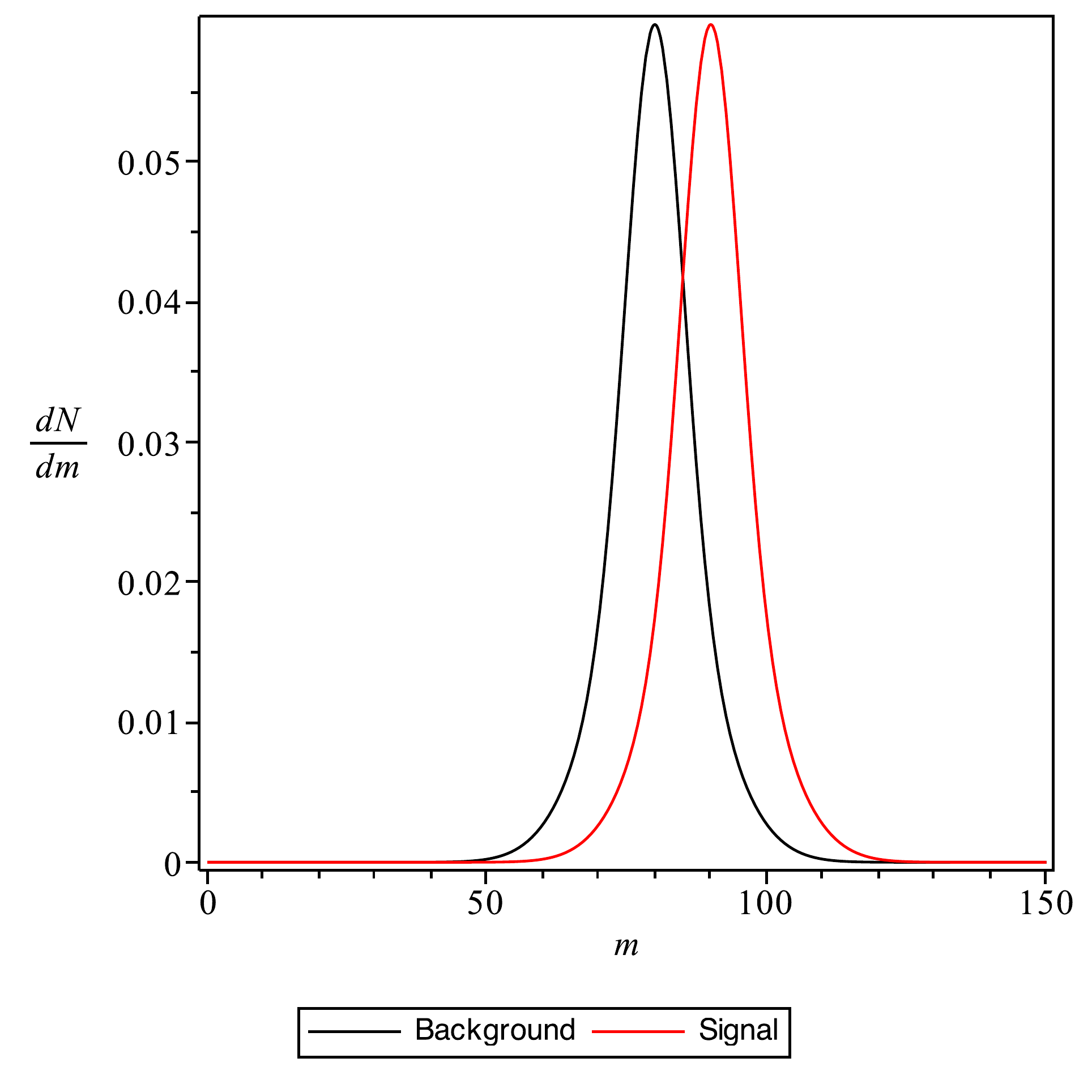}\includegraphics[scale=0.25]{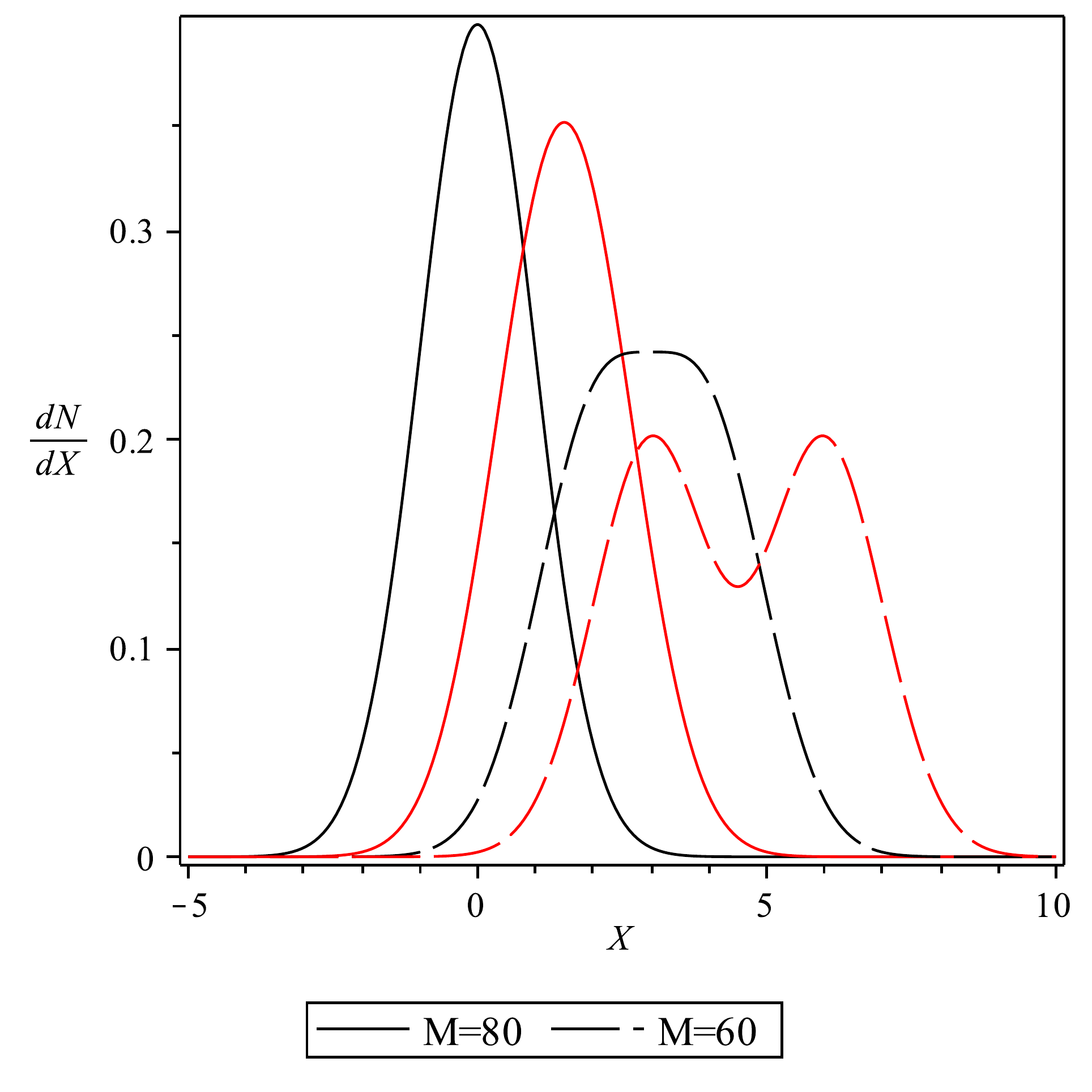}\includegraphics[scale=0.25]{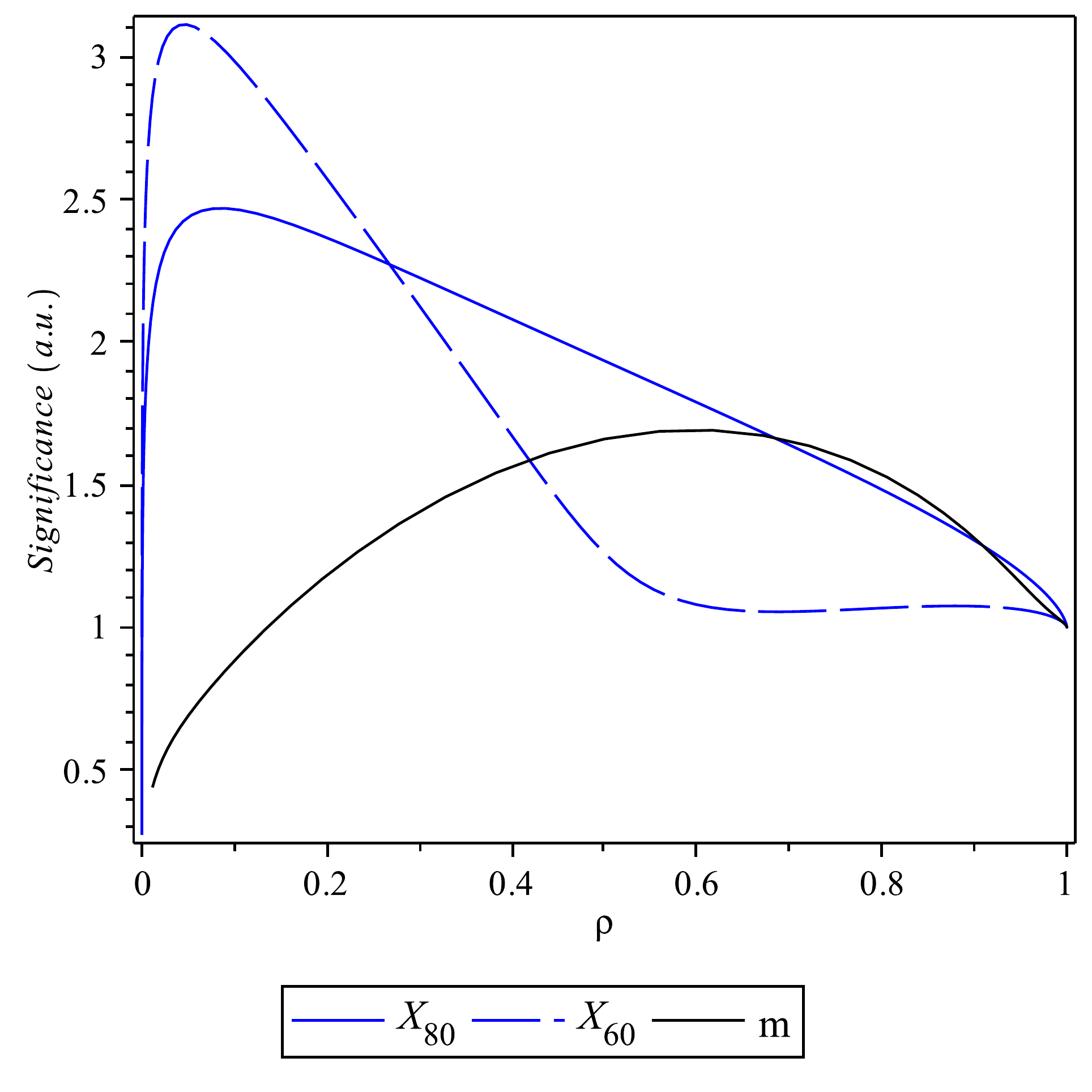}
\end{center}
\caption{These plots illustrate the distributions of $m$, $X$ and $\hat{s}$ for a simple model in which $m$ is always `on shell' at 80 for the background and 90 for the signal.  The resolutions can take one of two values with probability $1/2$, independent of the physics process.}
\label{mstar}
\end{figure}

For further properties of $X$ and related variables, including a discussion of computation, see Appendix~\ref{props}.  

\section{Performance in fully simulated examples of physical interest}
\label{simulate}

Using \textsc{Pythia} 8.170~\cite{Pythia8,Pythia,PythiaSUSY}, we simulate the distributions of $X_M$\footnote{We do not show $P_M$ because we are assuming Gaussian resolution functions and thus $X_M$ captures all the information in $P_M$.  Furthermore, as noted in Appendix~\ref{props}, $P_M$ is very expensive to compute in the tails of the distributions, which are the most important regions for searches for new physics.  The variables $Q_M$ and $Y_M$ (c.f. Appendix~\ref{props}) require model dependance and are in general more involved to compute and we find in the cases we examined that there is not significant benefit over $X_M$.
} in canonical searches that use the variables $m=m_T$ and $m=m_{T2}$.  

\subsection{$W'$ (new gaugue boson), transverse mass significance}

The transverse mass $m_T$ was first used in the discovery of the $W$ boson and the measurement of its mass at CERN by the UA1 collaboration~\cite{UA1}.   Defined by~\ref{Eqmt}, $m_T$ has the property that $m_T\leq m_W$.   Since its first use, $m_T$ continues to be used for precise measurements of the $W$ boson mass, as well as in searches for new physics.  For example $m_T$ is actively in use to search for new gauge bosons like the $W'$~\cite{ATLASWprime, CMSWprime}.  We therefore use a $W'$ search with $m_T$ as a model system to construct the {\it transverse mass significance}.   We concentrate our attention on the leptonic $W/W'$ decays so that the resolution function is determined almost entirely by the resolution in the missing momentum vector.   In this search, the $W$ mass is a natural choice for $M$ in constructing $X_M$.   In our Monte Carlo study, we simulate $pp$ collisions at $\sqrt{s}=14$ TeV.   The W' boson is created with a mass\footnote{excluded by~\cite{ATLASWprime, CMSWprime}, useful here for illustration only} of 100 GeV and the same CKM matrix as the Standard Model W boson.  The resolution of the missing momentum was modeled as $\sigma_{\slashed{E}_T^{\text{}}}^{x,y}=0.5\sqrt{\sum E_T}$, where $\sum E_T$ is the sum of all visible momentum and follows the measured spectra in dijets~\cite{ScalingATLAS}.   The distributions of $m_T,X_M$ and $\hat{s}$ are shown in Fig.~\ref{W}.  The various rows of Fig.~\ref{W} demonstrate the affect of the $W$ width on the efficacy of $X_M$.  We can see that for a vary narrow resonance background, $X_M$ is much better than $m_T$, but as the width becomes large, the advantage decreases.  

\begin{align}
\label{Eqmt}
m_T^2=m_{\nu}^2+m_{\text{lepton}}^2+2\left(
\sqrt{m_{\nu}^2+
\slashed{E}_T^2
}
\sqrt{m_{\text{lepton}}^2+(p_T^{\text{lepton}})^2}-\slash\hspace{-2.5mm}\vec{E}_T^{\text{}}\cdot \vec{p}_T^{\text{   lepton}}\right)
\end{align}

\begin{figure}
\begin{center}
\includegraphics[scale=0.8]{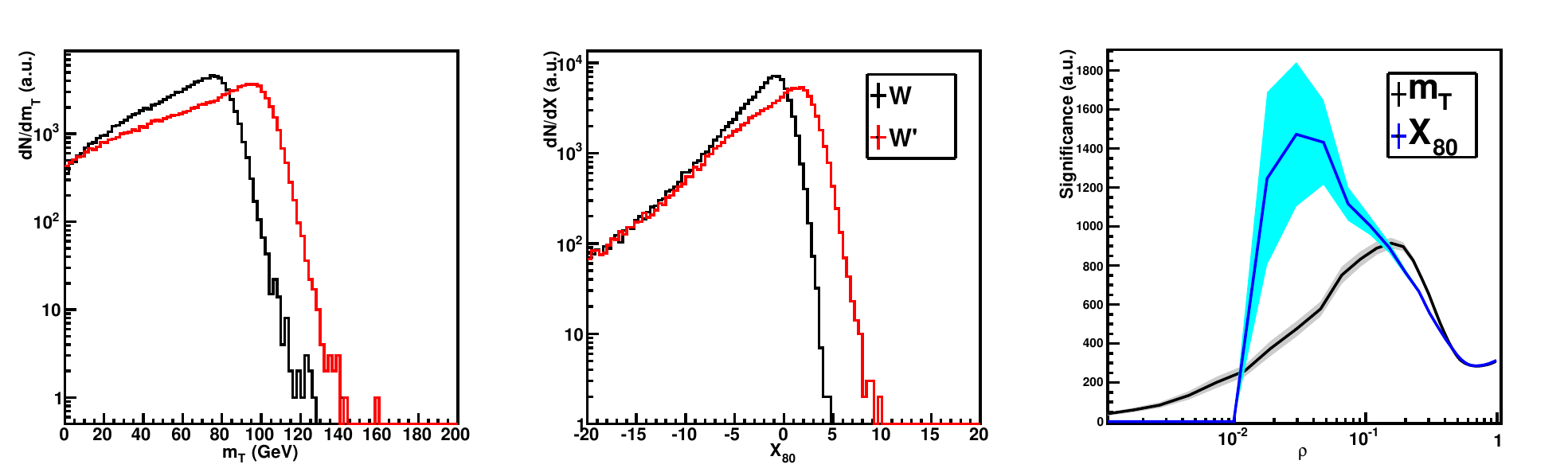}\\
\includegraphics[scale=0.8]{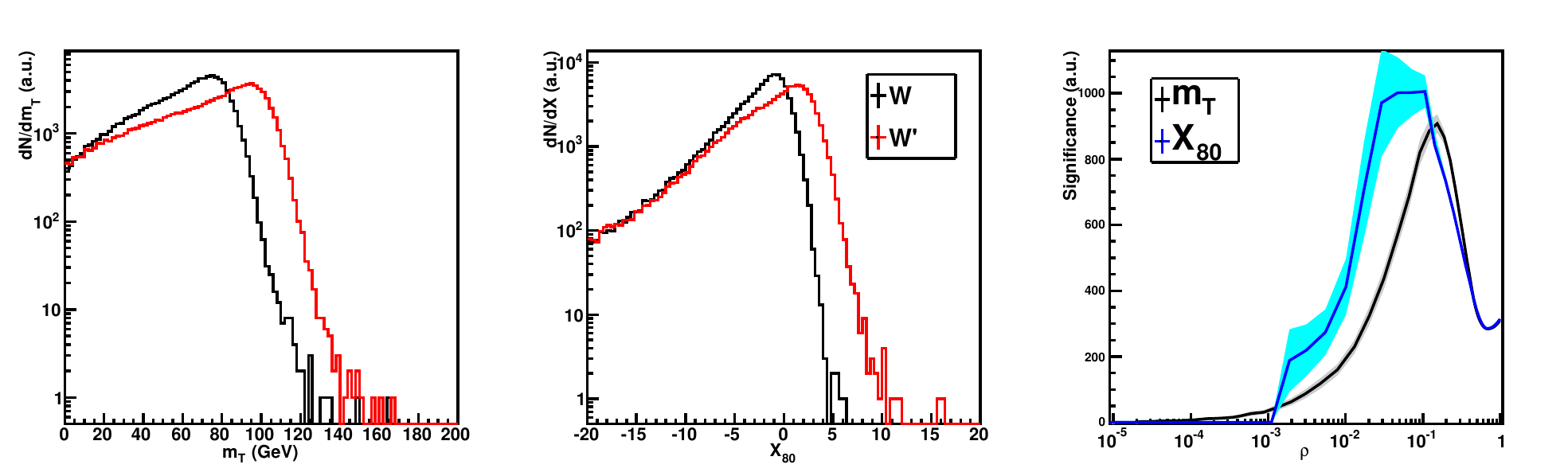}\\
\begin{overpic}[scale=0.8]{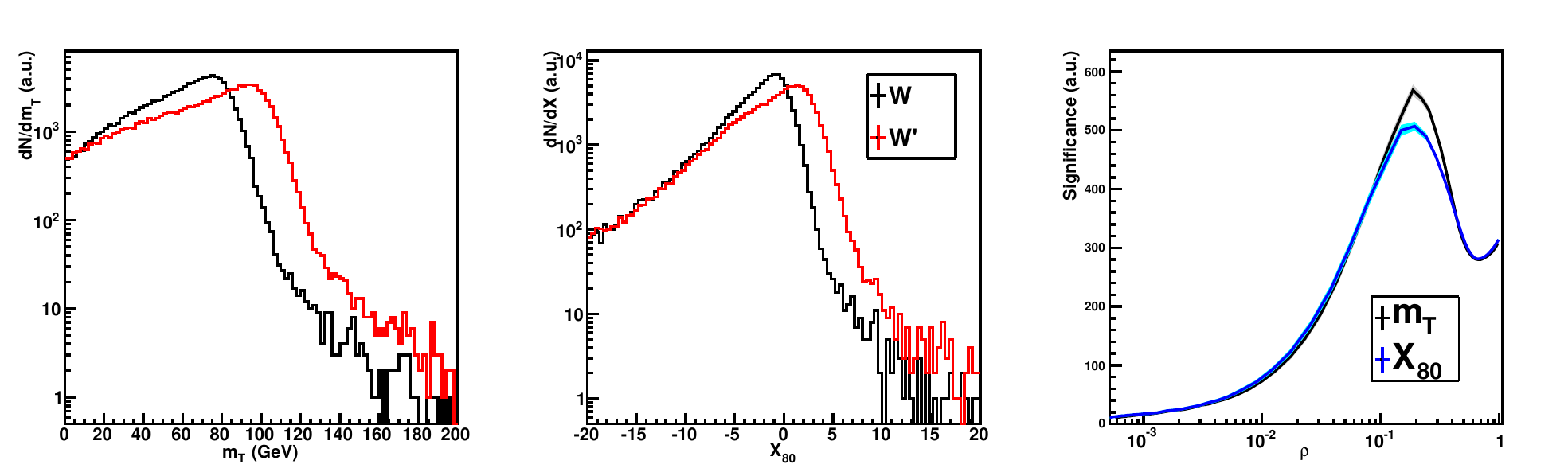}\put(71,14){\includegraphics[scale=0.39]{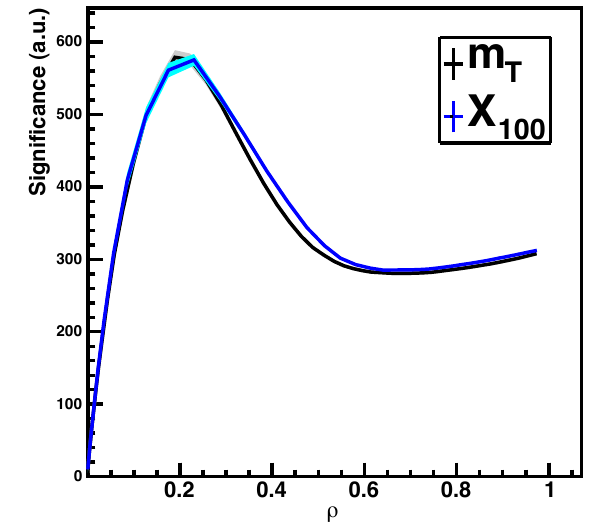}}\end{overpic}
\end{center}
\caption{In each row, the left plot compares the transverse mass distribution for a Standard Model W and a W' with mass 100 GeV.  The middle plot is the corresponding distributions of $X_M$ with $M=80$ GeV. The right plot shows the rejection $s\sqrt{b}$ as a function of the signal efficiency, in arbitrary units.  The bands show the statistical uncertainty due to limited Monte Carlo statistics.  The top row has a boson mass width of 0, the middle has a width of 20\%, and the bottom row has the full width.  We can see that for this fixed value of $M$, the performance of $X_M$ is better than $m_T$ for a narrow width and then worse at higher width.  By construction, $X_M$ cannot be worse than $m_T$ and thus the optimal $M$ in the last row must be different than 80.  The inset plot shows $X_M$ for $M=100$, for which the performance of $X$ and $m_T$ is the same.}
\label{W}
\end{figure}

\subsection{$H\rightarrow\tau\tau$, transverse mass significance}

Another possible use of the $m_T$ significance is in the standard $H\rightarrow\tau\tau$ search (measurement)~\cite{HtautauATLAS,HtautauCMS}.  In the dilepton channel, the dominant background is Z boson production and so the natural value for $M$ is $90$ GeV.  Figure~\ref{Higgs} shows the distributions of $m_T$ (between the total missing transverse momentum and the two lepton composite system), $X_M$, and $\hat{s}$ for a 125 GeV Higgs.  The optimal value of $M$ was found to be less than $90$, as indicated in the diagram.  The $\hat{s}$ figure shows that there can be a significant improvement from $X_M$ over $m_T$.

\begin{figure}
\begin{center}
\includegraphics[scale=0.8]{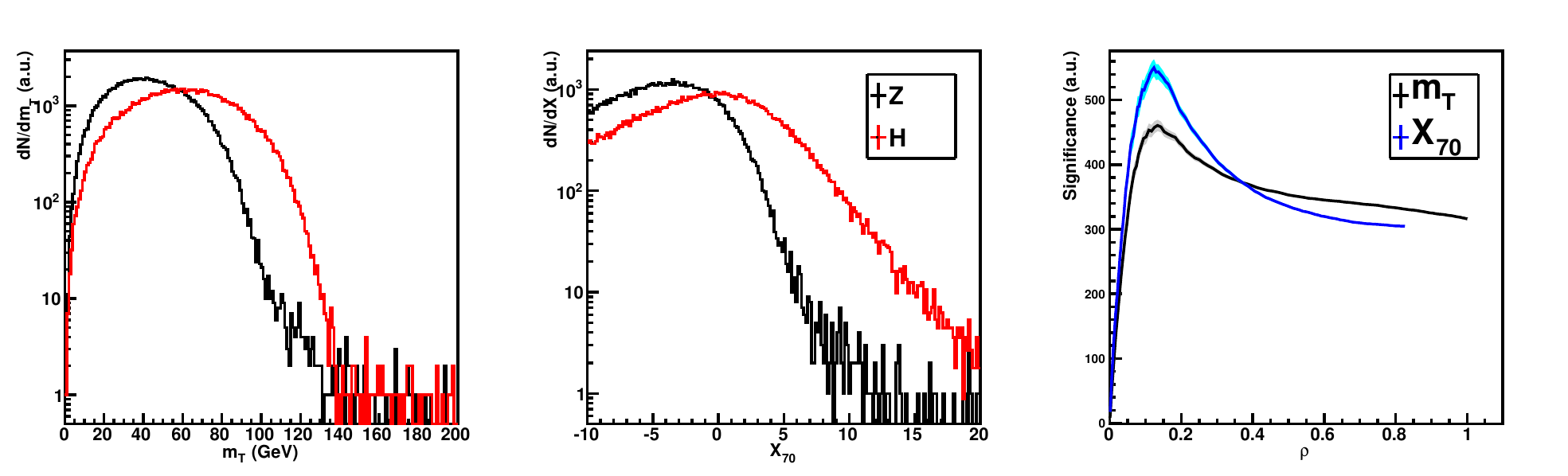}
\end{center}
\caption{The left plot is the $m_T$ distribution for dileptonic $Z\rightarrow\tau\tau$ and $H\rightarrow \tau\tau$  for a 125 GeV Higgs.  The middle plot is the corresponding X curve with M=60 and the right plot is the rejection versus efficiency relationship.}
\label{Higgs}
\end{figure}

\subsection{Pair production of light stops, $pp\rightarrow \tilde t \tilde t X$,  stransverse mass significance}

The transverse mass is very effective when there is one missing particle in an event topology, such as a neutrino.  However, with pair production of missing particles, additional considerations are required.  One natural generalization of $m_T$ is the variable $m_{T2}$~\cite{FirstDef}, defined by Eq.~\ref{eqMT2} for a symmetric event topology involving one visible particle and one missing particle in each branch.  The missing particle in branch $\chi\in\{a,b\}$ has transverse momentum $p_{T\chi}$ and $m_{T\chi}$ is the transverse mass of one branch formed by the corresponding missing particle momentum and the measured visible particle momentum.  Further generalizations of the $m_{T2}$ variable have been studied and applied to Tevatron and LHC data for mass measurements and searches for new physics.  For example, consider direct stop squark production in $R$-parity conserving SUSY.  There is a lot of interest now at the LHC in searches for these signatures for light stop squarks with all the other sparticles very heavy.  One such search in ATLAS uses $m_{T2}$ in the dileptonic channel~\cite{dileptonstopATLAS}.  It is this model that we use as our testing ground to construct the {\it stransverse mass significance}.  With the leptons as the visible particles in the definition of $m_{T2}$, this system once again has the feature that the resolution is mostly due to the missing momentum vector.   Since $t\bar{t}$ is the dominant background, we take $M=80$ GeV.  Here, we only consider the decay $\tilde{t}\rightarrow t+\mathrm{LSP}$.    The $m_{T2}$ distribution, stransverse mass significance, and $\hat{s}$ are shown in Fig.~\ref{Stop} for a compressed scenario of $m_{stop}=350$ GeV and $m_{LSP}=170$ GeV.

\begin{equation}
\label{eqMT2}
 m_{T2}\equiv\underset{\vec{p}_{Ta}^C+\vec{p}_{Tb}^C=\hspace{1mm}\slash\hspace{-2mm} \vec{E}_{T}}{\min} \{\max (m_{Ta},m_{Tb})\}
\end{equation}

\begin{figure}
\begin{center}
\includegraphics[scale=0.8]{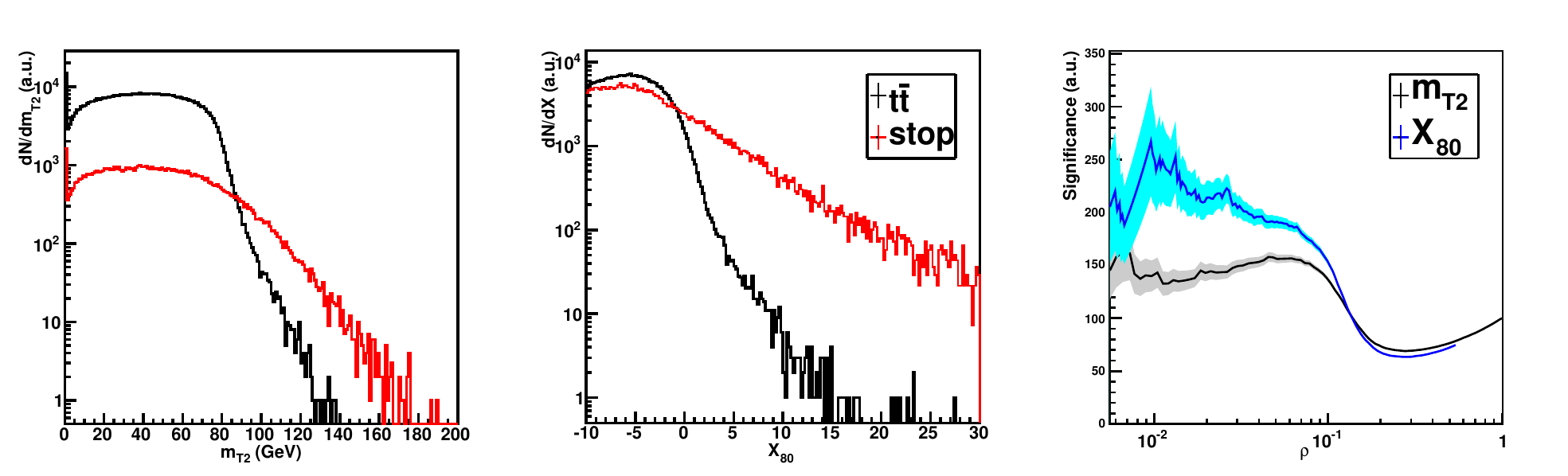}
\end{center}
\caption{The left plot is the $m_{T2}$ distribution for for dileptonic $t\bar{t}$ and $\tilde{t}\rightarrow t+\mathrm{LSP}$ for a 350 GeV stop and 170 GeV LSP.  The middle plot is the corresponding X curve with M=80 and the right plot is the rejection versus efficiency relationship.}
\label{Stop}
\end{figure}

\section{Conclusions}

Given any bounded kinematic variable $m$, we have constructed the significance variable $X_M$ and its variants $Y_M, P_M$ and $Q_M$ which generalize the idea of missing transverse momentum significance.  We have proved that (for an appropriate choice of $M$) the significance variable $X_M$, alone, {\em cannot} perform worse than the variable $m$ upon which it is based.  We have found concrete and physically interesting examples of the significance variable $X_M$ performing better than $m_T$ or $m_{T2}$ as a discrimination variable.  In particular, for $H\rightarrow\tau\tau$ we find that $X_M$ can outperform $m_T$ with respect to  $s/\sqrt{b}$ by $\sim20\%$ and for direct stop production $X_M$ is better than $m_{T2}$ by $\sim 30\%$.

\vspace{5mm}

Even though we have seen improvements from $X_M$ in some standard applications of bounded kinematic variables, the main purpose of this paper is to make a case that event-by-event resolution information should be included in all analyses.  The $X_M$-like significance variables provide a simple algorithm that may capture most of the relevant discriminatory information. When $X_M$ is not (nearly) optimal, the resolution information should be integrated into analyses with an MVA or a dedicated derivation of the optimal significance variable(s) for the analysis in question.   We hope that significance variables will now become part of the experimentalists standard toolbox.  

\section{Acknowledgments}

This work was supported by the Winston Churchill Foundation of America and by Peterhouse, Cambridge.  We thank members of the Cambridge Supersymmetry Working Group for useful discussions.

\appendix

\section{Properties of $X_M$ and related variables}
\label{props}

To begin, we need to show that the significance variable $X_M$ does indeed add new information over a search using $m$ alone.  This is not obvious, since it is often the case that the resolution of $m$ is uncorrelated with the underlying physics process.  In other words, the distribution of $\sigma(m)$ is the same for both signal and background.  Therefore, on its own $\sigma(m)$ does not provide any useful information.  To quantify the statement that $X_M$ adds new information, we can show that if some event $i$ lands in the tail region of $m$, it need not be in the tail region of $X_M$. 

\begin{prop}
For $N$ events, if $m$ induces an ordering on the events such that $m_1<m_2<\cdots <m_N$, then it is not necessarily true that $X_M^{(1)}<X_M^{(2)}<\cdots <X_M^{(N)}$. 
\end{prop}
 
\begin{proof}
 We can show this simply by demonstrating the $M$ dependance of $X_M$.  It is easiest to see when $N=2$ and to view $X_M$ as a function of $M$.  There are two possible configurations, as illustrated in Figure~\ref{Mdependance}.  In $(X_M,M)$ space, $X_M$ is a linearly decreasing function of $M$.   The quantity which controls the ordering of $X_M$ is $\Delta \equiv (m_2\sigma_1 -  m_1\sigma_2)/(\sigma_1-\sigma_2)$.   When $\Delta<0$ or infinite in magnitude, then $X_M^{(1)}>X_M^{(2)}$ for all $M$.  However, if $\Delta>0$, then there is a critical $M^*$ such that for $M<M^*$, $X_M^{(1)}>X_M^{(2)}$ for $M>M^*$, $X_M^{(1)}<X_M^{(2)}$.  The value of $M^*$ is $\Delta$.  For $N>2$, the situation is more complicated, but the result is the same; different values of $M$ can rearrange the distribution of events based on $X_M$ from the distribution based on $M$.   One can generalize the plots in Figure~\ref{Mdependance} for $N>2$.  Note that the distribution of points of intersection with the $M$ axis forms the observed distribution of $m$.
\end{proof} 

		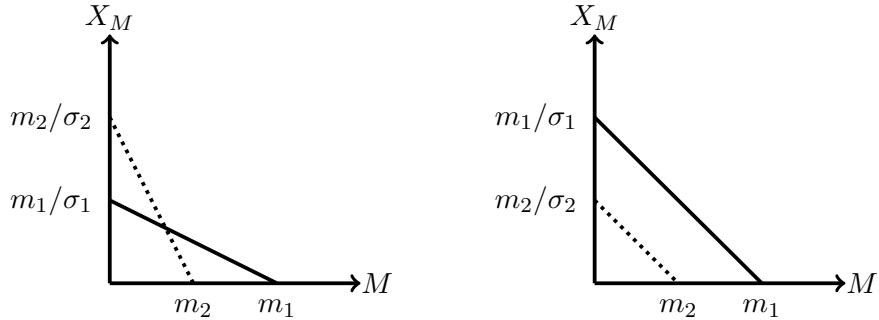
\begin{figure}
         	\begin{center}
		\begin{tikzpicture}[scale=1.1]
		\draw[->,thick,line width=0.5mm] (0,0.0) -- (0,3);
		\draw[->,thick,line width=0.5mm] (0,0.0) -- (3,0);
		\draw[thick,line width=0.5mm] (0,1) -- (2,0);
		\draw[thick,line width=0.5mm,dotted] (0,2) -- (1,0);
		\node at (3.2,0) { $M$};
		\node at (2,-0.3) { $m_1$};
		\node at (1,-0.3) { $m_2$};
		\node at (-.7,2) { $m_2/\sigma_2$};
		\node at (-0.7,1) { $m_1/\sigma_1$};
		\node at (0,3.2) { $X_M$};
		\begin{scope}[shift={(5.8,0)}]
		\draw[->,thick,line width=0.5mm] (0,0.0) -- (0,3);
		\draw[->,thick,line width=0.5mm] (0,0.0) -- (3,0);
		\draw[thick,line width=0.5mm] (0,2) -- (2,0);
		\draw[thick,line width=0.5mm,dotted] (0,1) -- (1,0);
		\node at (3.2,0) { $M$};
		\node at (0,3.2) { $X_M$};
		\node at (2,-0.3) { $m_1$};
		\node at (1,-0.3) { $m_2$};
		\node at (-0.7,2) { $m_1/\sigma_1$};
		\node at (-0.7,1) { $m_2/\sigma_2$};
		\end{scope}
		\end{tikzpicture}
		\end{center}
		\caption{The dependance of $X_M$ on $M$ for two events with $\Delta\equiv (m_2\sigma_1 -  m_1\sigma_2)/(\sigma_1-\sigma_2) > 0$ in the left plot and $\Delta\in [-\infty,0)\cup\{\infty\}$ in the right plot.}
		\label{Mdependance}
		\end{figure}

Now, we return to the original motivation for constructing a new variable from $m$.  We observed that in the absence of detector resolution, $m$ has a kinematic maximum $M$.  If we let $m^{\text{true}}$ denote the value of $m$ that we would observe given a delta function response function from the detector, then this means that the probability that $\Pr(m^{\text{true}}>M)=0$.  We therefore are motivated to try to compute the probability that $m^{\text{true}}>M$ for a given event since this is zero for the Standard Model background.  However, since we do not know the true value, the best we can do is compute 

\begin{align}
Q_M\equiv \Pr(m^{\text{true}}>M|m^{\text{observed}}).
\end{align}

At first, it may seem like $Q_M$ and $P_M$ (from Eq.~\ref{defofP}) contain the same information, but in fact this is not the case.  
  
  \begin{prop}
  If $P_M$ induces an ordering on $N$ events given by $P_M^{(1)}<P_M^{(2)}<\cdots <P_M^{(N)}$, then it is not necessarily the case that $Q_M^{(1)}<Q_M^{(2)}<\cdots <Q_M^{(N)}$.  
  \end{prop}

\begin{proof}  
  To see this, consider the case in which $N=2$.  Then, we can compute the difference $Q_M^{(1)}-Q_M^{(2)}$ and relate it to $P_M^{(1)}-P_M^{(2)}$.   Even in the case in which $R$ is a Gaussian, the quantity:
\begin{align}
\label{DeltaQ}
Q_M^{(1)}-Q_M^{(2)}&=\int_M^\infty \left[p(m^{\text{true}}|m_1^{\text{observed}})-p(m^{\text{true}}|m_2^{\text{observed}})\right]dm^{\text{true}}\\\nonumber
&=\frac{1}{p(m_1^{\text{observed}})}\int_M^\infty \left[p(m_1^{\text{observed}}|m^{\text{true}})-\frac{p(m_1^{\text{observed}})}{p(m_2^{\text{observed}})}p(m_2^{\text{observed}}|m^{\text{true}})\right]p(m^{\text{true}})dm^{\text{true}}.
\end{align}
is not necessarily positive given that $P_M^{(1)}-P_M^{(2)}$ is positive.  In this case, $X_M^{(1)}-X_M^{(2)}$ determines $\left[p(m_1^{\text{observed}}|m^{\text{true}})-p(m_2^{\text{observed}}|m^{\text{true}})\right]$.  However, because the ratio of probabilities multiplying the second term in the second line of Eq.~\ref{DeltaQ} could be important and since $X_M$ has $M$ dependance, the integral does not just depend on the values of $p(m_i^{\text{observed}}|m^{\text{true}})$ at the endpoints $\{m,\infty\}$ due to the weighting function $p(m^{\text{true}})$.  
\end{proof}

  Just as we formed $X_M$ out of $P_M$ (Eq.~\ref{defofP}), we could form a variable $Y_M$ from $Q_M$ of the form 
  
  \begin{align}
  Y_M\equiv \frac{m^{\text{observed}}-M}{\sigma_m[R']},
  \end{align}
  
   where $R'=p(m^{\text{true}}|m^{\text{observed}})$.  In the case that $R'$ is a Gaussian, this completely determines the behavior of $Q_M$ in the sense that both $Q_M$ and $Y_M$ induce the same ordering of events.  However, due to falling prior distributions, it is not often the case that $R'$ is exactly Gaussian, though $Y_M$ is still useful because it is easier to compute than $Q_M$.  Even though both $Q_M$ and $Y_M$ aim to probe the truth structure of an event, one drawback is that they both require knowledge of the prior $p(m^{\text{true}})$.  We cannot get this distribution from the observed data, instead relying on Monte Carlo simulations.

\section{Computation of $X_M$, $Y_M$, $P_M$ and $Q_M$}
\label{compute}

  First, we consider the Gaussian variable $X_M$.  Jet and lepton responses are parametrized as a function of their coordinates in $(\eta,p_T)$ space.  This response is defined to be the ratio $p_T^{\text{observed}}/p_T^{\text{true}}$ so we have access to the variance of $p(p_T^{\text{observed}}|p_T^{\text{true}})$.  For $X_M$, however, we would like to know the width of $p(p_T^{\text{(re)measured}}|p_T^{\text{observed}})$.  For ease of notation, let $\rho=p_T^{\text{(re)measured}},\mu=p_T^{\text{measured}},\tau=p_T^{\text{true}}$.  Using the law of total probability and Bayes' Law, we can expand $p(\rho|\mu)$ as in Eq.~\ref{totalprob}.   

\begin{align}
\label{totalprob}
p(\rho|\mu)&=\int p(\rho|\mu,\tau)p(\tau|\mu)d\tau\\\nonumber
&=\int p(\rho|\tau)p(\tau|\mu)d\tau\\\nonumber
&=\int p(\rho|\tau)\frac{p(\mu|\tau)p(\tau)}{p(\mu)}d\tau\nonumber
\end{align}

 Now, suppose that we know the prior distribution $p(\tau)$ in terms of a histogram: $p(\tau)=\sum\alpha_i\delta_i(\tau)$ where $i=1,...,N$ is the number of bins and $\delta_i$ is the indicator function on the bin $i$ over range $[a_i,b_i]$.  Then, in Eq.~\ref{manipulate}, we insert this function into the results from Eq.~\ref{totalprob}.   In Eq.~\ref{manipulate}, $\text{Gauss}(x,\mu,\sigma)$ is a Gaussian with mean $\mu$ and standard deviation $\sigma$ evaluated at $x$. 

\begin{align}
\label{manipulate}
p(\rho|\mu)p(\mu)&=\sum_i\alpha_i\int_0^\infty p(\rho|\tau)p(\mu|\tau)\delta_i(\tau)d\tau\\\nonumber
&=\sum_i\alpha_i\int_{a_i}^{b_i}p(\rho|\tau)p(\mu|\tau)d\tau\\\nonumber
&=\sum_i\alpha_i\int_{a_i}^{b_i}\text{Gauss}(\rho,\tau,\sigma)\text{Gauss}(\mu,\tau,\sigma)\\\nonumber
&=\sum_i\alpha_i\text{Gauss}(\rho,\tau,\sqrt{2}\sigma)\left[\text{erf}\left(\frac{2a_i-\rho-\mu}{2\sigma}\right)-\text{erf}\left(\frac{2b_i-\rho-\mu}{2\sigma}\right)\right]\\\nonumber
&=\text{Gauss}(\rho,\tau,\sqrt{2}\sigma)\sum_i\alpha_i\left[\text{erf}\left(\frac{2a_i-\rho-\mu}{2\sigma}\right)-\text{erf}\left(\frac{2b_i-\rho-\mu}{2\sigma}\right)\right]\\\nonumber
&:=\text{Gauss}(\rho,\tau,\sqrt{2}\sigma)\sum_i\alpha_i(*),\\\nonumber
\end{align}

 Now, we want to understand how $(*)$ in Eq.~\ref{manipulate} varies with $\rho$, since we view $p(\rho|\mu)$ as fixed in $\mu$ and as a function of $\rho$.  In Eq.~\ref{derivative}, we observe that the dependance of $(*)$ in Eq.~\ref{manipulate} on $\rho$ goes to zero as $a_i\rightarrow b_i$ and thus to a good approximation, $p(\rho|\mu)\propto\text{Gauss}(\rho,\mu,\sqrt{2}\sigma)$.  Practically then, to compute $X_M$, one must propagate these `inflated' Gaussians into a formula for the resolution function of $m$.

\begin{align}
\label{derivative}
\frac{d(*)}{d\rho}\propto \text{Gauss}(2b_i,\rho+\mu,\sqrt{2}\sigma)-\text{Gauss}(2a_i,\rho+\mu,\sqrt{2}\sigma).
\end{align}

 If the Gaussian approximation for the resolution function is very good, then an analytic approximation using linear error propagation would be sufficient.  However, to capture non-Gaussian attributes, numeric propagation may be necessary.  In particular, if $m$ is a mass-like variable with a restriction $m>0$, the resolution function will necessarily be non-Gaussian near $m=0$.  In such cases, we can estimate how many random draws are necessary to accurately compute $\sigma_m$.   If $s^2$ is the sample variance, then the variance of the sample variance is given by Eq.~\ref{varofsamplevar}, where $\kappa$ is the excess kurtosis~\cite{VarSampleVar}.  

\begin{align}
\label{varofsamplevar}
\text{Var}[s^2]=\sigma^4\left(\frac{2}{n-1}+\frac{\kappa}{n}\right).
\end{align}

\noindent For an absolute uncertainty on the standard deviation $f$ and an $\mathcal{O}(1)$ standard deviation, one needs

\begin{align}
n=\frac{2+\kappa+f^2+\sqrt{4+4\kappa+4f^2+\kappa^2-2f^2\kappa+f^4}}{2f^2}.
\end{align}

\noindent For $f\ll 1$ and an order $1$ or smaller $\kappa$ (this is zero for a Gaussian),

\begin{align}
n\approx\frac{2+\kappa+\sqrt{4+4\kappa+\kappa^2}}{2f^2}\sim\frac{3}{f^2}.
\end{align}

\noindent For example, one needs $n\approx 300$ for an accuracy of $0.1$ GeV.  The computation for $Y_M$ is similar to $X_M$, except instead of propagating the uncertainties from $p(\rho|\mu)$, one must propagate the uncertainties from $p(\tau|\mu)$, which requires the input of a prior distribution $p(\tau)$.  In general, these priors are expected to not be uniform and thus the propagation must be done numerically as linear error propagation may not be accurate. 

The computation of $P_M$ and $Q_M$ may seem must harder than that of $X_M$ and $Y_M$.  However, this may not be the case.  To ease the notation, we recycle letters from earlier by letting $\rho=m^{\text{(re)measured}} $and $\mu=m^{\text{measured}}$.  Then, we can rewrite $P_M$ as in Eq.~\ref{Prewrite}, where $\Theta$(x) is the Heavyside step function and the expectation value in the last line is taken over the space with measure given by the conditional distribution $p(\rho|\mu)$.

\begin{align}
\label{Prewrite}
P_M &:= \Pr(\rho>M | \mu)\\\nonumber
 &= \int \Pr( \rho>M | \mu,\rho) p(\rho|\mu) d\rho \\\nonumber
 &= \int \Pr( \rho>M | \rho) p(\rho|\mu) d\rho \\\nonumber
 &=  \int \Theta(\rho-M) p(\rho|\mu) d\rho\\\nonumber
 &= \left\langle \Theta(\rho-M)\right\rangle_{(\rho|\mu)}.
 \end{align}
 
 The reason for the different representation of $P_M$ in Eq.~\ref{Prewrite} is that it gives rise to an intuitive method for its computation and an easy way to assess its uncertainty.  Since $\Theta(x)\in\{0,1\}$, we can think of the expectation above as a Bernoulli random variable.  If the real value of $P_M$ is $p$, then the variance of the sample mean is $p(1-p)/n$ and thus the uncertainty is on the order of $\sqrt{p(1-p)/n}$.  For an absolute uncertainty $f$ on the mean $p$, then 

\begin{align}
\label{scaling}
n=\frac{p(1-p)}{f^2}\geq \frac{0.5(1-0.5)}{f^2}=\frac{1}{4f^2}.
\end{align}

 For example, one needs $n\approx 2500$ for an absolute uncertainty of $0.01$.   We make a similar computation for $Q_M$ and note that $Q_M=\left\langle \Theta(m^{\text{true}}-M)\right\rangle_{(m^{\text{true}}|m^{\text{observed}})}$.  The uncertainty bound on $Q_M$ is thus similar to $P_M$, except that one must input truth distributions when sampling.  In order to meaningfully compare $X_M$ and $P_M$, one needs a way of relating a given uncertainty on $X_M$ to an uncertainty on $P_M$.  We can do this quantifying the interpretation of $X_M$ as a `number of standard deviations beyond the endpoint', by using a map $\lambda:\mathbb{R}\rightarrow [0,1]$ given by Eq.~\ref{map}.  Given $\lambda$, we can ask how uncertainties in $X_M$ translate to uncertainties in $\lambda(X_M)$, which we can take as the necessary level of precision needed on $P_M$.  Figure~\ref{lambda} shows the relationship between $\sigma X_M$ and $\sigma\lambda$ for several values of $X_M$.  We can see that if $X_M\sim 1$, then a 10\% uncertainty in $X_M$ corresponds to $\sim 0.05$ absolute uncertainty in $P_M$.  However, if $X_M\sim 4$ then a $0.1$ absolute uncertainty on $X_M$ ($\sim 3\%$) then the required uncertainty on $P_M$ is $\sim 
10^{-5}$.  Since the absolute scaling of $X_M$ and $P_M$ is the same, this shows that it is very expensive to compute $P_M$.  Even though $P_M$ can encode non-Gaussian features of resolution functions, the computation cost may not outweigh the benefit from the computationally cheap $X_M$.

\begin{align}
\label{map}
\lambda(X_M)=\int_{-X}^X\mathrm{d}x\text{ Gauss}(x,0,1).
\end{align}

\begin{figure}
\begin{center}
\includegraphics[scale=0.5]{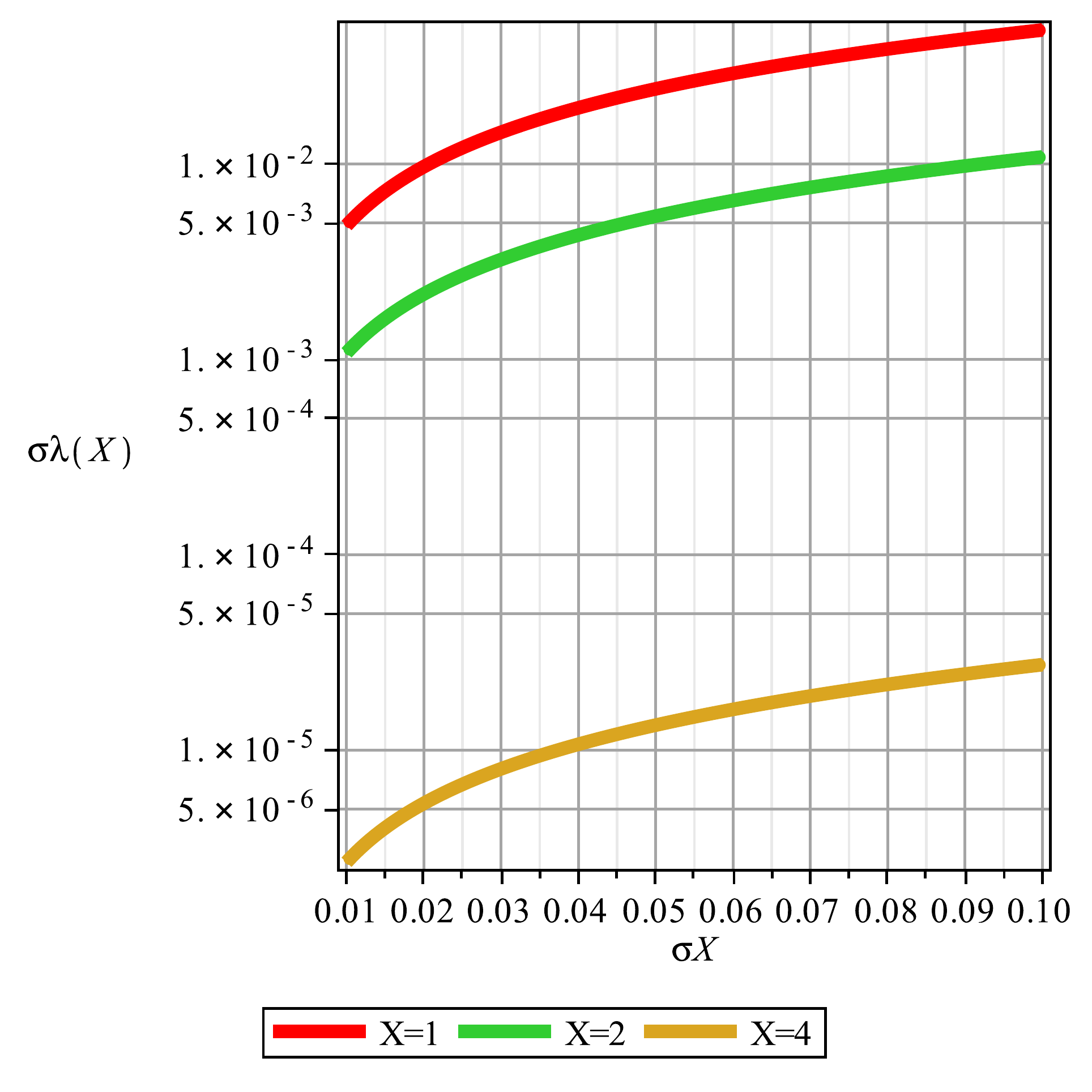}
\end{center}
\caption{Using the map $\lambda$ between significance and probability, we can relate the absolute uncertainty on $X_M$ to an absolute uncertainty on $\lambda(X_M)$, which is the precision we would need on $P_M$ to make a meaningful comparison.}
\label{lambda}
\end{figure}

\section{Optimum use of additional variables}

\label{sec:proofsandstuff}

\newcommand\ee{{\sigma}}
\newcommand\vx{{\bf x}}
\newcommand\xv{{\vx}}
\newcommand\vm{{\bf m}}
\newcommand\thet{{\Theta}}
\newcommand\sigsm{{\sigma}}
\newcommand\sigmin{{\sigma_{\rm min}}}

The conclusions of this appendix on the optimal use of variables are not new.  However, it may be useful to review what appears in the literature to be `common knowledge.' Assume an event is characterized by an observable $x$ and an uncertainty $\ee$.  In other words once an event is recorded, values for $x$ and $\ee$ would be immediately known.  Note that below, $x$ and $\ee$ are treated simply as variables with a joint distribution $p(x,\ee)$ with no particular use made of the concept of $\ee$ as an uncertainty on a measurement made by the other, though that interpretation is possible within the framework.  Let $\vx = \{x,\ee\}$ and consider an arbitrarily function $f(\vx)$ which (in effect) defines a new variable.  For example, $X=\frac{x-M}{\ee}$ is an example of such a function, this time containing a parameter $M$.

\par

Consider two processes $s$ (signal) and $b$ (background) that we want to distinguish.  Signal events have a joint probability density function of the form $p_s(\vx)$, background events $p_s(\vx)$, and the mixture of both has distribution: $p(\vx,\lambda)=\lambda p_s(\vx) + (1-\lambda) p_b(\vx)$ where $\lambda\in [0,1]$ is the fraction of signal events.

\par

Given the processes $s$ and $b$, we can construct many functions $f$ and consider an analysis $A_f$ which takes $N_T$ total events and selects a subset $N\leq N_T$ for which $f\ge 0$.  For each analysis, we can construct a measure of performance by computing the expected value (with respect to $p$) of some optimality metric $K(N_s,N_b)$ where $N_s+N_b=N$ and $N_s$ is the number of true signal events of the $N$ selected by $A_f$.  For example, $K=N_s/\sqrt{N_b}$ is a standard metric.  An analysis $A_f$ is optimal with respect to $K$ if no other choice of $f$ produces a higher value of $K$.   Optimal choices of $f$ are not unique -- we can take an optimal analysis $A_f$ and transform $f$ by wrapping it within any function $g$ that maps non-negative values to non-negative values and maps negative values to negative values and produce the same analysis and thus the same $K$.  The important parts of $f$ are therefore (i) its zeros (which define the boundary between accepted and rejected events) and (ii) its sign as a function of $\vx$.  We will see this fact (re)emerge from the mathematics later.

\par

Hereafter take $f(\vx)$ to be an optimal choice of $f$ for some $K$, and create a (possibly non-optimal) function $g(\vx,\mu)=f(\vx) + \mu h(\vx)$ where $h(\vx)$ is an arbitrary polluting function of $\vx$ and $\mu$ is a scalar parameter controlling the degree of non optimality of $g$. Clearly $g$ becomes optimal when $\mu=0$.   Let
\begin{equation}D_i(\mu)=\int \thet(g(\vx,\mu)) p_i(\vx) d \vx,\label{eq:D}\end{equation}
\noindent for $i\in \{s,b\}$ and $\thet$ is the Heaviside step function.  With this definition, the expected number of signal and background events for $N$ events total in an analysis using the possibly non-optimal discriminant $g(\mu)$ are given by $N_s = N \lambda D_s$ and $N_b = N (1-\lambda) D_b$, and so if $K$ were to take the explicit form $K_\text{example}\equiv N_s/\sqrt{N_b}$ then we would have
$$K^2(\mu) = \frac{N^2 \lambda^2}{N (1-\lambda)} \frac{(D_s(\mu))^2}{D_b(\mu)}.$$

\par

Since $g$ is optimal when $\mu=0$ we know that $\frac{\partial K^2}{\partial \mu} =0$ when evaluated at $\mu=0$, independent of the choice of $h(\vx)$.  Accordingly, a necessary condition for optimality of $f$ (assuming that $N$ is non-zero and that $\lambda$ is neither zero nor one) is
$$ 1 D_b(0) D'_s(0) -{\frac 1 2} D_s(0) D'_b(0)=0$$
in the case that $K=K_\text{example}$, or for arbitrary $K$ would take the form
\begin{equation} \kappa_s D'_s(0) +\kappa_bD'_b(0)=0\label{eq:Constraint}\end{equation}
in which $\kappa_i\equiv \left.\frac{\partial K}{\partial D_i}\right|_{\mu=0}$.
Now we compute
\begin{equation}
D'_i(\mu) = \int \delta(f(\vx) + \mu h(\vx)) p_i(\vx) h(\vx) d\vx,
\end{equation}
\noindent and note that we have freedom to choose any $h(\vx)$.  We exercise that freedom by making the choice $h(\vx) = \delta^{(n)}(\vx - \vm)$ for some and arbitrary constant $\vm$, where $n$ is the dimension of our $\vm$ space.  With this particular choice of $h(\vx)$,  Eq.~\ref{eq:Constraint} becomes:
$$
\kappa_s\delta(f(\vm)) p_s(\vm) + 
\kappa_b\delta(f(\vm)) p_b(\vm)=0,
$$
or equivalently
\begin{equation}
\left[\delta(f(\vm))\right]  \times\left[ \kappa_s  p_s(\vm) +
\kappa_b p_b(\vm) \right] = 0, \label{eq:factors}
\end{equation}
which must be true for any choice of $\vm$.  The presence of the two separate terms (multiplied together) in Eq.~\ref{eq:factors} reminds us of our earlier statements about which parts of $f$ should matter.  For one thing, it shows us that for all values of $m$ which are off the boundary defined by $f(\vm)=0$ the first term (containing the delta function) is zero, and so off of this boundary, there are no special constraints on $f$ deriving from $\kappa_s$, $\kappa_b$, $p_s$ and $p_b$.  These parameters are only relevant insofar as they affect {\em the location of} the optional boundary $f(\vx)=0$.  We see that this optimal boundary is therefore controlled exclusively by the second of the two terms in Eq.~\ref{eq:factors} and its equality to zero.  The boundary determining condition from the second term alone can be re-written as the requirement
\begin{equation}
\frac{p_s(\vm)}{p_b(\vm)} = -\frac{\kappa_b}{\kappa_s}, \label{eq:strong}
\end{equation}
which (we recall) must be satisfied by {\it all} values of $\vm$ which lie on the optimal boundary $f(\vm)=0$.  In particular, the lefthand side of Eq.~\ref{eq:strong} is a function of $\vm$ whereas the righthand side is not!  Accordingly, the values of $\vm$ that occupy the boundary must be exactly those for which $$\rho(\vm) = \frac{ p_s(\vm)}{ p_b(\vm)}$$ is a constant and equal to $-\kappa_b/\kappa_s$.   Effectively, therefore, we now have all we need to know to construct the optimal $f(\vx)$. All we need to do is the following:
\begin{enumerate}
\item Consider the 1-parameter family of curves in the $\{x,\ee\}$-plane that satisfy $\rho(\vx)=\frac{p_s(\vx)}{p_b(\vx)}=const=\rho$, and consider them to be indexed by this real parameter $\rho$.
\item Treat each curve as defining a boundary between two regions of the plane, these regions being named $R^+_\rho$ and $R^-_\rho$ respectively.
\item Let $R=\left\{ R^+_\rho | \rho\in \mathbb{R} \right\} \bigcup \left\{R^-_\rho | \rho \in \mathbb{R} \right\}$ be the set of all such regions.
\item
For each region $r\in R$ calculate the fraction of signal events $F_s(r)$ expected to fall within $r$: 
$$F_s(r) = \int_r p_s(\vx) d\vx$$ and calculate the same quantity for background events:
$$F_b(r) = \int_r p_b(\vx) d\vx.$$
\item The optimal cut boundary $f(\vx)=0$ will be the boundary of the region $r\in R$ for which $F_s(r)/F_b(r)$ equals the value of $\rho$ which defined that region $r$.
\end{enumerate}


\end{document}